\RequirePackage{luatex85}

\documentclass[11pt,a4paper]{article}
\usepackage{iftex}

\ifPDFTeX
\usepackage[utf8]{inputenc}
\usepackage[T1]{fontenc}
\fi

\usepackage{amsmath}
\usepackage{amssymb}
\usepackage{amsthm}
\usepackage{mathtools}
\usepackage{thmtools}
\usepackage{isomath}

\ifPDFTeX
\usepackage{libertine}
\usepackage[libertine]{newtxmath} 
\usepackage[scaled=0.96]{zi4} 
\else
\usepackage{fontspec}
\usepackage{unicode-math}
\fi

\usepackage[DIV=10]{typearea} 

\usepackage{microtype}

\usepackage[ruled,vlined,linesnumbered]{algorithm2e}

\usepackage{xcolor}

\usepackage[textwidth=2cm,textsize=footnotesize]{todonotes}
\usepackage{lineno}

\usepackage[
	style=alphabetic,
    maxalphanames=4,
    minalphanames=4,
    maxcitenames=4,
    minnames=1,
    maxbibnames=20,
	backref=true,
	doi=true,
	url=true,
	backend=bibtex8,
]{biblatex}

\usepackage[ocgcolorlinks]{hyperref} 

\newcommand*\samethanks[1][\value{footnote}]{\footnotemark[#1]}

\allowdisplaybreaks[1]

\makeatletter
\g@addto@macro\bfseries{\boldmath}
\makeatother

\makeatletter
\g@addto@macro\mdseries{\unboldmath}
\g@addto@macro\normalfont{\unboldmath}
\g@addto@macro\rmfamily{\unboldmath}
\g@addto@macro\upshape{\unboldmath}
\makeatother

\DeclareCiteCommand{\citem}
    {}
    {\mkbibbrackets{\bibhyperref{\usebibmacro{postnote}}}}
    {\multicitedelim}
    {}

\renewcommand*{\multicitedelim}{\addcomma\space}

\AtEveryCitekey{\clearfield{note}}

\makeatletter
\AtBeginDocument{%
    \newlength{\temp@x}%
    \newlength{\temp@y}%
    \newlength{\temp@w}%
    \newlength{\temp@h}%
    \def\my@coords#1#2#3#4{%
      \setlength{\temp@x}{#1}%
      \setlength{\temp@y}{#2}%
      \setlength{\temp@w}{#3}%
      \setlength{\temp@h}{#4}%
      \adjustlengths{}%
      \my@pdfliteral{\strip@pt\temp@x\space\strip@pt\temp@y\space\strip@pt\temp@w\space\strip@pt\temp@h\space re}}%
    \ifpdf
      \typeout{In PDF mode}%
      \def\my@pdfliteral#1{\pdfliteral page{#1}}
      \def\adjustlengths{}%
    \fi
    \ifxetex
      \def\my@pdfliteral #1{}
      \def\adjustlengths{\setlength{\temp@h}{-\temp@h}\addtolength{\temp@y}{1in}\addtolength{\temp@x}{-1in}}%
    \fi%
    \def\Hy@colorlink#1{%
      \begingroup
        \ifHy@ocgcolorlinks
          \def\Hy@ocgcolor{#1}%
          \my@pdfliteral{q}%
          \my@pdfliteral{7 Tr}
        \else
          \HyColor@UseColor#1%
        \fi
    }%
    \def\Hy@endcolorlink{%
      \ifHy@ocgcolorlinks%
        \my@pdfliteral{/OC/OCPrint BDC}%
        \my@coords{0pt}{0pt}{\pdfpagewidth}{\pdfpageheight}%
        \my@pdfliteral{F}
        \my@pdfliteral{EMC/OC/OCView BDC}%
        \begingroup%
          \expandafter\HyColor@UseColor\Hy@ocgcolor%
          \my@coords{0pt}{0pt}{\pdfpagewidth}{\pdfpageheight}%
          \my@pdfliteral{F}
        \endgroup%
        \my@pdfliteral{EMC}%
        \my@pdfliteral{0 Tr}
        \my@pdfliteral{Q}%
      \fi
      \endgroup
    }%
}
\makeatother

\colorlet{DarkRed}{red!50!black}
\colorlet{DarkGreen}{green!50!black}
\colorlet{DarkBlue}{blue!50!black}

\hypersetup{
	linkcolor = DarkRed,
	citecolor = DarkGreen,
	urlcolor = DarkBlue,
	bookmarksnumbered = true,
	linktocpage = true
}

\DontPrintSemicolon
\SetKwProg{Procedure}{Procedure}{}{}
\SetProcNameSty{textsc}
\SetFuncSty{textsc}
\SetKw{KwAnd}{and}
\SetKw{KwDo}{do}

\SetCommentSty{mycommfont}

\declaretheorem[numberwithin=section]{theorem}
\declaretheorem[numberlike=theorem]{lemma}

\declaretheorem[numberlike=theorem]{corollary}
\declaretheorem[numberlike=theorem]{definition}

\newcommand{\mL}{\boldsymbol{\mathit{L}}}

\DeclareMathOperator{\poly}{poly}

\usetikzlibrary{shapes.geometric, arrows, positioning, shapes.multipart}

\DeclareMathOperator{\ID}{ID}

\DeclareMathOperator{\Z}{\mathbb{Z}}

\DeclareMathOperator{\R}{\mathbb{R}}

\DeclareMathOperator{\Vol}{Vol}

\DeclareMathOperator{\Omegat}{\widetilde{\Omega}}

\SetAlgoSkip{bigskip}

\title{The Laplacian Paradigm in Deterministic Congested Clique}

\author{
  Sebastian Forster\thanks{Department of Computer Science, University of Salzburg, Austria. This work is supported by the Austrian Science Fund (FWF): P 32863-N. This project has received funding from the European Research Council (ERC) under the European Union's Horizon 2020 research and innovation programme (grant agreement No~947702).}
  \and
  Tijn de Vos\samethanks
}

\date{}

\hypersetup{
	pdftitle = {The Laplacian Paradigm in Deterministic Congested Clique},
	pdfauthor = {Tijn de Vos, Sebastian Forster}
}

\begin{document}
\maketitle
\begin{abstract}
In this paper, we bring the techniques of the Laplacian paradigm to the congested clique, while further restricting ourselves to deterministic algorithms. In particular, we show how to solve a Laplacian system up to precision $\epsilon$ in $n^{o(1)}\log(1/\epsilon)$ rounds. We show how to leverage this result within existing interior point methods for solving flow problems. We obtain an $m^{3/7+o(1)}U^{1/7}$ round algorithm for maximum flow on a weighted directed graph with maximum weight $U$, and we obtain an $\tilde O(m^{3/7}(n^{0.158}+n^{o(1)}\poly\log W))$ round algorithm for unit capacity minimum cost flow on a directed graph with maximum cost~$W$. Hereto, we give a novel routine for computing Eulerian orientations in $O(\log n \log^* n)$ rounds, which we believe may be of separate interest. 

\end{abstract}

\newpage
\tableofcontents
\newpage

\section{Introduction}
Over the last years, the Laplacian paradigm has proven itself successful in both sequential and parallel settings. To discuss this paradigm, let us first define the \emph{Laplacian matrix} $L(G)\in \R^{n\times n}$ of a undirected $n$-node graph $G$: $L(G):= D(G)-A(G)$, where $D(G)$ is the diagonal weighted degree matrix, and $A(G)$ is the weighted adjacency matrix. If $w\colon E \to \R$ is the weight function, we have 
\begin{align*}
    L(G)_{u,v} = \begin{cases}
        \sum_{(x,u)\in E}w(u,x) & \text{if }u=v;
        \\
        -w(u,v) &\text{if }(u,v)\in E;
        \\
        0 & \text{else.}
    \end{cases}
\end{align*}
A \emph{Laplacian system} is an equation $L(G)x=y$, and a \emph{Laplacian solver} solves, given $L(G)$ and $y\in \R^n$, a Laplacian system for $x\in \R^n$. The Laplacian paradigm encompasses the study of spectral techniques, among which Laplacian solvers, for solving graph problems. This was initiated by Spielman and Teng~\cite{SpielmanT04}, who presented a near-linear time Laplacian solver. Subsequently, more efficient sequential and parallel solvers have been developed~\cite{KOSZ13,KoutisMP14,KoutisMP11, CohenKMPPRX14,KyngS16,PengS14,KLP+16}. This line of research has found many applications, including but not limited to flow problems~\cite{Madry13,Sherman13,KelnerLOS14,Madry16,Peng16,CMSV17,LiuS2020FasterDivergence,LiuS20,AxiotisMV20}, bipartite matching~\cite{BrandLNPSS0W20}, and (distributed/parallel) shortest paths~\cite{BKKL17,Li20,AndoniSZ20}.

Recently the Laplacian paradigm has been brought to distributed models. In particular, it has appeared in the CONGEST 
 model~\cite{BKKL17,GhaffariKKLP18,FGLP+20,Anagnostides0HZ22} and in the Broadcast Congested Clique~\cite{BKKL17,ForsterV22}. We continue this line of work by bringing the Laplacian paradigm to the more powerful congested clique~\cite{LPSPP05}, but restrict ourselves to \emph{deterministic} algorithms. 
 
 Our first contribution is an algorithm for solving Laplacian systems in the congested clique.\footnote{The notation used in \autoref{thm:Laplacian_detCC} is reviewed in \autoref{sc:prelim}.}
\begin{restatable}{theorem}{ThmLaplacian}\label{thm:Laplacian_detCC}
There is a deterministic algorithm in the congested clique that, given an undirected graph $G=(V,E)$, with positive real weights $w\colon E\to \R_{\geq 0}$, bounded by $U$, and a vector $b\in \R^n$, computes a vector $x\in \R^n$ such that $||x-L_G^\dagger b||_{L_G} \leq \epsilon ||L_G^\dagger b||_{L_G}$ in $n^{o(1)}\log(U/\epsilon)$ rounds. 
\end{restatable}
 
We apply our Laplacian solver to the maximum flow problem and the unit capacity minimum cost flow problem, again using only deterministic algorithms.\footnote{Throughout, we absorb $\poly\log U$ factors in $n^{o(1)}$ whenever $U$ appears polynomially as well.} 
\begin{restatable}{theorem}{ThmFlow}\label{thm:maxflow_detCC}
There exists a deterministic algorithm that, given a graph $G=(V,E)$ with integer capacities $u\colon E\to \{1,2,\dots, U\}$, solves the maximum flow problem in $m^{3/7+o(1)}U^{1/7}$ rounds in the congested clique. 
\end{restatable}

\begin{restatable}{theorem}{ThmMinCostFlow}\label{thm:mincostflow_detCC}
There exists a deterministic algorithm that, given a graph $G=(V,E)$ with unit capacities and integer costs $c\colon E\to \{1,2,\dots, W\}$, and a demand vector $\vec{\sigma}\colon V \to \Z$, solves the minimum cost flow problem in $\tilde O(m^{3/7}(n^{0.158}+n^{o(1)}\poly\log W))$ rounds in the congested clique. 
\end{restatable}

Using the other techniques presented in this paper, but replacing the Laplacian solver by a simpler, randomized solver (see~\cite{ForsterV22}), we can convert the $n^{o(1)}$ in both flow theorems above into a $\poly\log n$ factor. 

As a subroutine for both flow algorithms, we develop an algorithm that computes Eulerian orientations. The problem is, given a \emph{Eulerian graph}, i.e., a graph where each vertex has even degree, to compute a \emph{Eulerian orientation}, i.e., an orientation of the edges such that each vertex has equally many edges entering as exiting.

\begin{restatable}{theorem}{ThmEulerianOrientation}\label{thm:EulerianOrientation}
    There exists a deterministic congested clique algorithm that, given a Eulerian graph $G=(V,E)$, finds a Eulerian orientation in $O(\log n \log^* n)$ rounds. 
\end{restatable}

\subsection{Comparison to Related Work}

\paragraph{Distributed Laplacian Solvers}
Forster, Goranci, Liu, Peng, Sun, and Ye~\cite{FGLP+20} gave an almost existentially optimal Laplacian solver in the CONGEST model. They solve a Laplacian system up to precision $\epsilon$ in $n^{o(1)}(\sqrt n+D)\log(1/\epsilon)$ rounds, and they give a lower bound of $\Omegat(\sqrt n+D)$ for any $\epsilon\leq \tfrac{1}{2}$. Subsequently, Anagnostides, Lenzen, Haeupler, Zuzic, and Gouleakis~\cite{Anagnostides0HZ22} showed that this is not \emph{instance} optimal. They provide a solver that takes $n^{o(1)}\poly(SQ(G))\log(1/\epsilon)$ rounds for a graph $G$, where $SQ(D)$ is the shortcut quality of $G$, see \cite{GhaffariH16} for a definition. For a range of graphs (including planar, expander, and excluded-minor graphs) they provide an $n^{o(1)}SQ(G)\log(1/\epsilon)$ round algorithm. This is almost \emph{universally} optimal: for a graph $G$, they show a lower bound of $\Omegat(SQ(G))$, for any $\epsilon\leq \tfrac{1}{2}$.

For the Broadcast Congested Clique, Forster and de Vos~\cite{ForsterV22} obtained an algorithm in the $\poly \log n$ regime: they give a randomized algorithm taking $O(\poly\log(n)\log(1/\epsilon))$ rounds.

\paragraph{Distributed Flow Algorithms}
In the distributed setting, there are three exact flow algorithms at the time of writing; two for the CONGEST model~\cite{FGLP+20} and one for the Broadcast Congested Clique~\cite{ForsterV22}. In the CONGEST model, the maximum flow algorithm takes $\tilde O(m^{3/7}U^{1/7}(n^{o(1)}(\sqrt n+D)+\sqrt n D^{1/4})+\sqrt m)$ rounds and the unit capacity minimum cost flow algorithm takes $\tilde O(m^{3/7+o(1)}(\sqrt nD^{1/4}+D)\poly\log W )$ rounds. In the Broadcast Congested Clique, the minimum cost flow algorithm takes $\tilde O(\sqrt n)$ rounds.  Clearly all these algorithms can be implemented in the congested clique, as both other models are more restrictive. The CONGEST algorithms are clearly always slower than ours. Essentially, the outer procedures are the same (the interior point method algorithms of M\k{a}dry~\cite{Madry16} and Cohen, M\k{a}dry, Sankowski, and Vladu~\cite{CMSV17} respectively), and we keep the number of iterations $m^{3/7}U^{1/7}$ and $\tilde O(m^{3/7})$ respectively, but we bring down the cost of each iteration to $n^{o(1)}$ and $\tilde O(n^{0.158}+n^{o(1)}\poly\log W)$. The Broadcast Congested Clique algorithm is faster than our algorithms for sufficiently dense graphs. However, both these algorithms are randomized: both Laplacian solvers are randomized, and in the Broadcast Congested Clique algorithm also uses randomization in its interior point method itself.

For the deterministic setting, to the best of our knowledge, there are only two algorithms against which we should compare ourselves. First of all, there is the trivial $O(n\log U)$ round algorithm that makes all knowledge global, such that the problem can be solved internally at each node. 

Secondly, we compare ourselves against the Ford-Fulkerson algorithm~\cite{Ford1956maximal}. Let $|f^*|$ denote the value of the maximum flow, then the Ford-Fulkerson algorithm takes $|f^*|$ iterations, where each iteration consists of solving an $s$-$t$ reachability problem. We can solve such a problem in $O(n^{0.158})$ rounds~\cite{Censor-HillelKK19}, giving us a total round complexity of $O(|f^*|n^{0.158})$. Note that the value of the maximum flow $|f^*|$ can be as big as $nU$, while this is only more efficient than the trivial algorithm above when $|f^*|=o(n^{0.842}\log U)$. 

For completeness, let us include here the CONGEST model approximate maximum flow algorithm by Ghaffari, Karrenbauer, Kuhn, Lenzen, and Patt-Shamir~\cite{GhaffariKKLP18}. They give an $n^{o(1)}(\sqrt n+D)/\epsilon^3$ round algorithm for $(1+\epsilon)$-approximate maximum flow in weighted \emph{undirected} graphs. 

\paragraph{Flow Algorithms: Iteration Count and Round Complexity}
In recent years, interior point methods have booked many successes in solving flow problems, with most recently Chen, Kyng, Liu, Peng, Probst Gutenberg, and Sachdeva~\cite{ChenKLPPS22} solving minimum cost flow in $m^{1+o(1)}$ time. Their algorithm uses $m^{1+o(1)}$ iterations, where each iteration can be implemented in $m^{o(1)}$ (amortized) time. While this is efficient in a central setting, it does not seem suitable for a distributed setting. The reason hereto is that, unless they can be performed completely internally, each iteration consists of at least one round. Therefore, the algorithms with low iteration count are more suitable, such as the aforementioned algorithms of~\cite{Madry16,CMSV17,LS14}. The challenge here is to show that each iteration can be implemented efficiently in the distributed model, which we show is the case in the deterministic congested clique if each iteration essentially consists of a small number of Laplacian solves.

\paragraph{Eulerian Orientations}
To the best of our knowledge, there are only two results concerning Eulerian orientations in the distributed setting. Both are deterministic CONGEST model algorithms. The first is the $\tilde O(\sqrt m+D)$ round algorithm of Forster, Goranci, Liu, Peng, Sun, and Ye~\cite{FGLP+20}. The second is the improved $\tilde O(\sqrt n +D)$ round algorithm of Rozhon, Grunau, Haeupler, Zuzic, and Li~\cite{RozhonGHZL22}. Eulerian orientations have both been used as a subroutine for flow computation~\cite{FGLP+20} via the algorithm of Cohen~\cite{Cohen95}, and for shortest path computation~\cite{RozhonGHZL22}. Hence we believe that our congested clique algorithm may find applications beyond flow computation. 

We note that computing Eulerian orientations seems to be a hard problem in the Broadcast Congested Clique. The reason hereto is that each node can be part of many cycles, but for an efficient algorithm a node is only allowed to send designated messages for a few of them, which means there is no natural implementation of Cohen's algorithm. Hence, flow rounding also seems hard in the Broadcast Congested Clique. In their minimum cost flow algorithm for the Broadcast Congested Clique, Forster and de Vos~\cite{ForsterV22} avoid this problem by using the powerful tool of an LP solver. This solver gives an approximate solution so close to the optimal solution that internal rounding suffices. This technique has a $\tilde O(\sqrt{n})$ iteration count and uses randomized subroutines. As far as we know, computing Eulerian orientations is the bottleneck for implementing the flow algorithms of M\k{a}dry~\cite{Madry16} and Cohen, M\k{a}dry, Sankowski, and Vladu~\cite{CMSV17} in the Broadcast Congested Clique. Finally, let us note that as far as we know the fastest way to solve the Eulerian orientation problem in the Broadcast Congested Clique is to utilize the LP solver of Forster and de Vos for a $\tilde O(\sqrt n)$ round algorithm.

\subsection{Technical Overview}
In this section, we first describe our algorithm for a Laplacian solver, and for Eulerian orientations. The latter immediately leads to a flow rounding procedure. Together this gives the maximum flow algorithm and the minimum cost flow algorithm. 

\paragraph{A Laplacian Solver Through Spectral Sparsification}
Since we operate within a clique, we make use of the fact that once we have a spectral sparsifier, this can easily be made known to every node, since the graph is so sparse. Hence then using standard techniques, the Laplacian system can be solved (approximately) internally, and with $\log(1/\epsilon)$ iterations up to precision $\epsilon$. So the intermediate goal is to compute a deterministic spectral sparsifier. Hereto, we use the  work of Chang and Saranurak~\cite{CS20}, who show that we can efficiently compute a deterministic expander decomposition in the CONGEST model, so in particular in the congested clique. Then, we use the result from Chuzhoy, Gao, Li, Nanongkai, Peng, and Saranurak~\cite{CGL+20} that one can deterministically create a sparsifier by computing multiple expander decompositions. This technique is developed for a central setting, but we show it can be made to work in the congested clique as well. In total we obtain an $n^{o(1)}$-approximate spectral sparsifier in $n^{o(1)}$ rounds, and hence also an $n^{o(1)}\log(1/\epsilon)$ round $\epsilon$-approximate Laplacian solver. 

\paragraph{Flow Rounding and Eulerian Orientations}
Cohen~\cite{Cohen95} showed that flow rounding can be reduced to computing Eulerian orientations. We provide an $O(\log n\log^* n)$ round algorithm as follows. Each node starts by internally dividing its adjacent edges into pairs. This means we have an implicit decomposition into cycles. 

First, let us describe the procedure for a single such cycle. The goal is to reduce the cycle to fewer and fewer nodes, until only a leader remains. We do this by repeatedly selecting a constant fraction of the nodes on each cycle, and replacing the nodes in between by a new edge. Using randomization, one could simply sample each node with constant probability to obtain a set of selected nodes. For a deterministic algorithm, we use a $O(\log^* n)$ matching algorithm~\cite{CV86,GPS87}, and select the node with higher $\ID$ for each matched edge. In total, we have $O(\log n)$ iterations, where we select nodes in $O(\log^* n)$ rounds, and communicate in $O(1)$ rounds. At the end of this routine, one node is `aware' of the entire cycle -- it does not know it explicitly -- and decides on an orientation. Running the algorithm in reverse communicates this orientation to all nodes on the cycle. 

To perform this procedure on all cycles simultaneously, we use the routing algorithm of Lenzen~\cite{Lenzen13}, which shows that each communication step can be done efficiently in a constant number of rounds, even if some nodes are selected for many cycles. 

\paragraph{Maximum Flow}
For computing maximum flow, we use M\k{a}dry's algorithm~\cite{Madry16}. This is based on an interior point method with $\tilde O(m^{3/7}U^{1/7})$ iterations. In each iteration we solve an electrical flow instance. This can be done using the deterministic Laplacian solver mentioned above in $n^{o(1)}$~rounds. At the end of the interior point method, we have a high precision approximate solution. Using the above rounding algorithm, we round the fractional flow to an integral flow, which is almost optimal. It remains to find a single augmenting path to obtain the maximal flow.  This can be done in $O(n^{0.158})$~rounds using an algorithm from Censor-Hillel, Kaski, Korhonen, Lenzen, Paz, and Suomela~\cite{Censor-HillelKK19}. In the end we obtain an $m^{3/7+o(1)}U^{1/7}$ round algorithm for computing deterministic maximum flow, where $U$ is the maximum capacity of the graph. 

\paragraph{Unit Capacity Minimum Cost Flow}
For computing minimum cost flow in graphs with unit capacities, we use Cohen, M\k{a}dry, Sankowski, and Vladu's algorithm~\cite{CMSV17}. Just like the maximum flow algorithm above, this is based on an interior point method. The algorithm uses $\tilde O(m^{3/7}\poly\log W)$ iterations, where $W$ is the maximum cost. Each iteration can be solved using our deterministic Laplacian solver in $n^{o(1)}$~rounds. At the end of the interior point method, we have a low precision approximate solution: after rounding the flow to be integral using our algorithm mentioned above, we need $\tilde O(m^{3/7})$ augmenting paths to obtain the optimal flow. 
Each such path can be computed in $O(n^{0.158})$ rounds~\cite{Censor-HillelKK19}. In the end we obtain an $\tilde O(m^{3/7}(n^{0.158}+n^{o(1)}\poly\log W))$ round algorithm for computing deterministic unit capacity minimum cost flow.

\section{Preliminaries}\label{sc:prelim}
First we detail the model we work with. Next, we review spectral sparsifiers, and we show how they can be used for solving Laplacian systems. Finally, we formally introduce the maximum flow and minimum cost flow problems.

\subsection{The Computational Model}
The model we work with in this paper is called the \emph{congested clique}~\cite{LPSPP05}. It consists of a network of $n$ processors, which communicate in synchronous rounds. In each round, each processor (also called node) can send and receive a message of size $O(\log n)$ to/from each other processor over a non-faulty link. Each node has a unique identifier of size $O(\log n)$, initially only known by the node itself and its neighbors. Computational complexity is measured by the number of rounds a computation takes. At input to graph problems, nodes receive as input the edges incident to them, together with the corresponding capacities and costs.

There are two other models mentioned in this paper. First, there is the \emph{CONGEST model}~\cite{Peleg00}, where we have the additional restriction that nodes can only exchange message with their neighbors in the given network topology. 
Second, there is the \emph{Broadcast Congested Clique}~\cite{DKO12}, which is the congested clique with the additional constraint that in each round each node has to send \emph{the same message} to all other nodes. 

\subsection{Spectral Sparsification}\label{sc:prelim_spanners_sparsifiers}
Typical constructions of Laplacian solvers often involve the computation of a \emph{spectral sparsifier}, first introduced by Spielman and Teng~\cite{ST11}. This is a subgraph of which the Laplacian matrix approximates the Laplacian matrix of the input graph. Then solving the system using this matrix gives an approximate solution. More formally, we introduce spectral sparsifiers as follows. 
\begin{definition}
\label{def:specspars}
Let $G=(V,E)$ be a graph with weights $w_G\colon E\to \R$, and $n=|V|$, and let $\alpha\geq 1$ be a parameter. We say that a subgraph $H\subseteq G$ with weights $w_H\colon E(H)\to \R$ is an $\alpha$-approximate spectral sparsifier for $G$ if we have for all $x\in\R^n$:
\begin{align} \tfrac{1}{\alpha} x^TL_Hx \leq x^T L_G x\leq \alpha x^T L_Hx, \label{eq:specspars}\end{align}
where $L_G$ and $L_H$ are the Laplacians of $G$ and $H$ respectively.
\end{definition}

We introduce the short-hand notation $A \preccurlyeq B$ when $B-A$ is positive semi-definite. This reduces \autoref{eq:specspars} to $\tfrac{1}{\alpha}L_H \preccurlyeq L_G \preccurlyeq \alpha L_H$. 

Before we turn to Laplacian solving, let us give two more definitions. Given $x\in \R^n$ and $A\in \R^{n\times n}$, we define the norm of $x$ with respect to $A$ by $||x||_A:= \sqrt{x^T A x}$. Further, for $A\in \R^{n\times n}$ we write $A^\dagger$ for the \emph{pseudoinverse} of $A$. If we assume $\lambda_1, \dots, \lambda_r$ are the nonzero eigenvalues of $A$, with corresponding eigenvectors $u_1, \dots, u_r$, then $A^\dagger$ is defined by 
$$ A^\dagger := \sum_{i=1}^r \tfrac{1}{\lambda_i}u_i u_i^T.$$

\subsection{Laplacian Solving}\label{sc:laplacian}
We consider the following problem. Let $L_G$ be the Laplacian matrix for some graph $G$ on $n$ nodes. Given $b\in \mathbb R^n$, we want to solve $L_Gx=b$. Solving a Laplacian system exactly can be computationally demanding. 
Therefore, we consider an approximation to this problem: we want to find $y\in R^n$ such that $||x-y||_{L_G}\leq \epsilon||x||_{L_G}$. We compute such a $y$ by using a spectral sparsifier for $G$. We use the following theorem on \emph{preconditioned Chebyshev iteration}, a well known fact in numerical analysis~\cite{Saa03,axelsson_1994}. We use the formulation of Peng~\cite{peng13}.
\begin{theorem}[Preconditioned Chebyshev Iteration]
\label{thm:PCI}
There exists an algorithm \textsc{PreconCheby} such that for any symmetric positive semi-definite matrices $A, B$, and $\kappa$ where
\[A \preccurlyeq B \preccurlyeq \kappa A\]
any error tolerance $0\leq \epsilon \leq 1/2$, \textsc{PreconCheby}($A,B,b,\kappa$) in the exact arithmetic model is a
symmetric linear operator on $b$ such that:
\begin{enumerate}
\item\label{it:propZ} If $Z$ is the matrix realizing this operator, then $(1 -\epsilon)A^\dagger \preccurlyeq Z \preccurlyeq (1 + \epsilon)A^\dagger$;
\item For any vector $b$, \textsc{PreconCheby}($A,B, b,\epsilon$) takes $O(\sqrt{\kappa}\log (1/\epsilon))$ iterations, each
consisting of a matrix-vector multiplication by $A$, a solve involving $B$, and a constant number of vector operations.
\end{enumerate}
\end{theorem}

This yields the following corollary for Laplacian solving using spectral sparsifiers. 

\begin{corollary}
\label{cor:laplaciansolving}
Let $G$ be a weighted graph on $n$ vertices, let $\epsilon\in[0,1/2]$ be a parameter, and let $b\in \R^n$ a vector. Suppose $H$ is an $\alpha$-approximate spectral sparsifier for $G$, for some $\alpha \geq 1$. Then there exists an algorithm \textsc{LaplacianSolve}($G,H,b,\epsilon$) that outputs a vector $y\in \R^n$ such that $||x-y||_{L_G}\leq \epsilon||x||_{L_G}$, for some $x\in R^n$ satisfying $L_Gx=b$. \textsc{LaplacianSolve}($G,H,b,\epsilon$) takes $O(\sqrt{\alpha}\log(1/\epsilon))$ iterations, each consisting of a matrix-vector multiplication by $L_G$, a solve involving $L_H$, and a constant number of vector operations. 
\end{corollary}
\begin{proof}
As $H$ is an $\alpha$-approximate spectral sparsifier for $G$, we have: $\tfrac{1}{\alpha}L_H \preccurlyeq L_G \preccurlyeq \alpha L_H$, which we can rewrite to 
\[ L_G \preccurlyeq \alpha L_H \preccurlyeq \alpha^2 L_G.\]
We set $A:=L_G$ and $B:=\alpha L_H$, which are clearly both symmetric positive semi-definite. Furthermore, we set $\kappa:= \alpha$. We apply \autoref{thm:PCI} with these settings to obtain our vector $y=Zb$. Now we compare the result to $x=A^\dagger b$:
\begin{align*}
||x-y||_A &= \sqrt{(x-y)^T A(x-y)} \tag*{by definition of $||\cdot||_A$} \\
&= \sqrt{(A^\dagger b-Zb)^T A(A^\dagger b-Zb)}\\
&= \sqrt{b^T(A^\dagger A A^\dagger- A^\dagger A Z-ZAA^\dagger+ZAZ )b} \tag*{since $A^\dagger$ and $Z$ are symmetric} \\
&\leq \sqrt{b^T(A^\dagger A A^\dagger-2(1-\epsilon)A^\dagger A A^\dagger+(1+\epsilon)^2A^\dagger A A^\dagger)b} \tag*{by property~\ref{it:propZ} of $Z$}\\
&\leq \sqrt{\epsilon^2 x^TAx} \tag*{by the argument below}\\
&= \epsilon ||x||_A \tag*{by definition of $||\cdot||_A$.}
\end{align*}

It remains to show that $b^T(A^\dagger A A^\dagger-2(1-\epsilon)A^\dagger A A^\dagger+(1+\epsilon)^2A^\dagger A A^\dagger)b\leq \epsilon^2 x^TAx$. By property~\ref{it:propZ} of $Z$, we know that
$Z - (1-\varepsilon)A^\dagger \succeq 0$, $A \succeq 0$, and $(1+\varepsilon)A^\dagger - Z \succeq 0$ holds, thus, it follows that $AZA - (1-\varepsilon) A A^\dagger A\succeq 0$, $A^\dagger \succeq 0$, and
$(1+\varepsilon) A A^\dagger A - AZA \succeq 0$.

Observe that the product of the latter three matrices is
\begin{align*}
    &(1+\varepsilon) AZA A^\dagger A A^\dagger A - AZA A^\dagger AZA - (1-\varepsilon^2) A A^\dagger A A^\dagger A A^\dagger A + (1-\varepsilon) A A^\dagger A A^\dagger AZA \\
    &= 2 AZA - AZAZA - (1-\varepsilon^2) A A^\dagger A,
\end{align*}
which is symmetric, hence also positive semi-definite. Thus we have $\varepsilon^2 A^\dagger - (ZAZ - Z A A^\dagger - A^\dagger A Z + A^\dagger) \succeq 0$, which implies that 
\begin{align*}
    &b^T(A^\dagger A A^\dagger-2(1-\epsilon)A^\dagger A A^\dagger+(1+\epsilon)^2A^\dagger A A^\dagger)b = b^T (ZAZ - Z A A^\dagger - A^\dagger A Z + A^\dagger) b \\
    &\leq \varepsilon^2 b^T A^\dagger A A^\dagger b = \epsilon^2 x^TAx.\qedhere
\end{align*}

\end{proof}

\subsection{Flow Problems}\label{sc:prelim_flow}
In this section we formally define the maximum flow problem. Let $G=(V,E)$ be a directed graph, with integral capacities $u\colon E\to \Z_{\geq0}$, and designated source and target vertices $s,t\in V$. We say $f\colon E \to \R_{\geq 0}$ is a $s$-$t$ \emph{flow} if
\begin{enumerate}
	\item for each vertex $v\in V\setminus\{s,t\}$ we have $\sum_{e\in E : v\in E} f_e =0$;\label{item:flow1}
	\item for each edge $e\in E$ we have $f_e \leq u_e$. 
\end{enumerate}
The value of the flow $f$ is defined as $\sum_{u: (s,u)\in E}f_{(s,u)}$. The maximum flow problem is to find a flow of maximum value. 

Additionally, we can have costs on the edges: $c\colon E \to \Z_{\geq 0}$. The cost of the flow $f$ is defined as $\sum_{e\in E} c_ef_e$. The minimum cost (maximum) $s$-$t$ flow problem is to find a flow of minimum cost among all $s$-$t$ flows of maximum value. 

More generally, we can compute flows with respect to a demand vector $\sigma\colon V \to \Z$, satisfying $\sum_{v\in V} \sigma(v) =0$. Then we replace \autoref{item:flow1} by 
\begin{itemize}
    \item[(1')]  for each vertex $v\in V$ we have $\sum_{e\in E : v\in E} f_e =\sigma(v)$.
\end{itemize}
The minimum cost flow problem is then to find the flow of minimum cost among all feasible flows. This generalizes the minimum cost maximum $s$-$t$ flow, since we can binary search over the possible flow values for an $s$-$t$ flow.

\section{Deterministic Sparsifiers and Laplacian Solvers in the Congested Clique}
As outlined in \autoref{sc:laplacian}, we use spectral sparsifiers for our Laplacian solver. 
Chuzhoy, Gao, Li, Nanongkai, Peng, and Saranurak~\cite{CGL+20} show that one can deterministically create spectral sparsifiers from expanders. Using the deterministic expander decomposition construction from Chang and Saranurak~\cite{CS20}, we obtain a deterministic algorithm for spectral sparsifiers, and hence for solving Laplacian systems.

Let us start with the expander decomposition. Hereto, let us formally define conductance and expander decompositions. Expander decompositions were first introduced by Kannan, Vempala and Vetta~\cite{KVV04}, and have been a useful tool in a multitude of applications ever since. 

\begin{definition}
	Let $G=(V,E)$ be a graph. We define the \emph{conductance of $G$} by
\[ \Phi(G) := \min_{\emptyset\neq S \subsetneq V} \frac{|e_G(S,\overline{S})|}{\min\{\Vol(S),\Vol(\overline{S})\}},\]
where $e_G(S,\overline{S}):= \{(u,v)\in E : u\in S,v\in\overline{S}\}$ and $\Vol(S):=\sum_{v\in S}\deg_G(v)$. 
An \emph{$(\epsilon, \phi)$-expander decomposition} of $G$ is a partition $\mathcal{P}=\{V_1,\dots, V_k\}$ of the vertex set $V$, such that for all $i\in\{1,\dots, k\}$ we have $\Phi(G[V_i])\geq \phi$,\footnote{For $U\subseteq V$, we write $G[U]$ for the \emph{induced subgraph} by $U$, i.e., $G[U]:=(U, E\cap (U\times U ))$.} and the number of inter-cluster edges is at most an $\epsilon$ fraction of all edges, i.e., $|\{(u,v)\in E : u\in V_i,v\in V_j, i\neq j\}|\leq \epsilon |E|$. 
\end{definition}
The following theorem for computing expander decompositions holds in the CONGEST model, hence also holds on the congested clique. 
\begin{theorem}[\cite{CS20}]
\label{thm:exp_dec}
Let $0 < \epsilon < 1$ and $\sqrt{\log\log n/\log n}\leq \gamma < 1$ be parameters. An $(\epsilon,\phi)$-expander decomposition of a graph $G=(V,E)$ with $\phi = \epsilon^{O(1)}\log^{-O(1/\gamma)}n$ can be computed in $\epsilon^{-O(1)}n^{O(\gamma)}$ rounds deterministically. 
\end{theorem}

Next, we use this theorem to compute spectral sparsifiers in the congested clique.

\begin{theorem}\label{thm:detSS}
There is a deterministic algorithm that, given an undirected graph $G = (V, E)$ with weights $w\colon E\to\{1,2,\dots, U\}$, and a parameter $1 \leq r \leq  \sqrt{\log n/\log\log n}$, computes a $\log (n)^{O(r^2)}$-approximate spectral sparsifier $H$ for $G$, with $|E(H)| \leq O (n \log n \log U)$, in $O(\log n\log U n^{O(1/r^2)})$ rounds in the congested clique. At the end of the algorithm $H$ is known to every node. 
\end{theorem}
\begin{proof} 
We will show we can implement the algorithm of Chuzhoy, Gao, Li, Nanongkai, Peng, and Saranurak in the congested clique. For a complete proof of correctness, we refer to~\cite{CGL+20}. 

First, we assume $G$ is unweighted. We start by decomposing $G$ into expanders as follows. Let $\mathcal{P}_0$ be an $(1/2,\phi)$-expander decomposition of $G_0:=G$, with $\phi=1/\log(m)^{O(r^2)}$. This can be computed using \autoref{thm:exp_dec} in $O(\log m\epsilon^{-O(1)}n^{O(1/r^2)})$ rounds, by setting $\epsilon=1/2$ and $\gamma = O(1/r^2)$. Set $G_2\subseteq G$ to be the graph consisting of the edges crossing the expander decomposition $\mathcal{P}_1$. Now we repeat this procedure: find a $(1/2,\phi)$-expander decomposition $\mathcal{P}_i$ of $G_i$, and set $G_{i+1}\subseteq G$ to be the edges crossing the expander decomposition. As we have set $\epsilon=1/2$, we need to do this at most $O(\log m)=O(\log n )$ times. Hence we obtain $O(\log n)$ times a set of disjoint graphs of conductance at least $\phi$, in a total of $O(\log n n^{O(1/r^2)})$ rounds.

Next, we sparsify each of these graphs separately, and define the sparsifier of $G$ to be the union of all the separate sparsifiers. By linearity, we get the desired result. 

A graph $G'$ with conductance at least $\phi$, can be approximated by the \emph{product demand graph} $H(\deg_{G'})$. The product demand graph $H(\deg_{G'})$ is defined as a complete graph on $V(G')$, with weights $w(u,v)=\deg_{G'}(u)\cdot \deg_{G'}(v)$. We then have that $D:=\frac{2}{|E(G')|}H(\deg_{G'})$ is a $\frac{4}{\phi^2}$-approximate spectral sparsifier for $G'$~\cite{CGL+20}.

In the congested clique, constructing $D=\frac{2}{|E(G')|}H(\deg_{G'})$ just takes one round where each vector broadcasts its ID. This can be done for all expanders of $\mathcal{P}_i$ simultaneously, so this takes total time $O(\log n)$. 

Next, we sparsify $D$, which can be done internally at each node, see Kyng, Lee, Peng, Sachdeva, and Spielman~\cite{KLP+16} for an algorithm. To be precise, we compute a graph $D'$ that is a $2$-approximate spectral sparsifier of size $O(n)$. We conclude that, $D'$ is a $\frac{4}{\phi^2}$-approximate spectral sparsifier for $G'$. 

As the graphs in our decomposition have conductance at least $1/\log(m)^{O(r^2)}$, we obtain a $4\left(\log(m)^{O(r^2)}\right)^2=\log(m)^{O(r^2)}$-approximate spectral sparsifier. As each level of recursion gives a graph of size $O(n)$, we have total size $O(n\log n)$. Finally, to obtain the weighted result, it is sufficient to run the algorithm for binary weight classes, giving an extra factor of $\log U$ in the size and time. For details see Chuzhoy Gao, Li, Nanongkai, Peng, and Saranurak\cite{CGL+20}. 
\end{proof}

Next, we show how to use these approximate spectral sparsifiers to solve Laplacian systems. This is a standard argument that we include here for the sake of completeness. 
\ThmLaplacian*
\begin{proof}
First, we compute an $e^{O\left( (\log n \log\log n)^{1/2}\right)}$-approximate spectral sparsifier $H$ of $G$. This takes $O\left(\log(U/\epsilon)e^{O\left( (\log n \log\log n)^{1/2}\right)}\right)$~rounds, using \autoref{thm:detSS} with $r=\left(\frac{\log n}{\log\log n}\right)^{1/4}$. Note that we set the bound on the weights to $U/\epsilon$ instead of $U$, since \autoref{thm:detSS} takes integer weights. Hence we preprocess the graph (internally at each node) by rounding to the closest multiple of $\epsilon$, scaling by $1/\epsilon$, and after computing the sparsifier rescaling it by $\epsilon$. 

Next, we apply \autoref{cor:laplaciansolving}. This means we need $O(\sqrt{\alpha}\log(1/\epsilon))=e^{O\left( (\log n \log\log n)^{1/2}\right)}\log(1/\epsilon)$ iterations, consisting of a matrix-vector multiplication by $L_G$, a solve involving $L_H$, and a constant number of vector operations. Each such iteration can be performed in a constant number of rounds, which is clear for the matrix-vector multiplication and the vector operations. For the solve involving $L_H$, note that the sparsifier $H$ is known to each node, therefore any solve involving the sparsifier can be done internally at each node as well, so takes only one round. Finally we note that 
$$O(e^{O\left( (\log n \log\log n)^{1/2}\right)}(\log U+\log(1/\epsilon)))=O(n^{o(1)}(\log U+\log(1/\epsilon))),$$
which gives the result as stated
\end{proof}

\section{Flow Rounding}\label{sc:flow_rounding}
The interior point method used in this paper for computing maximum flow is a high-accuracy approximate solver. This means that the final answer is not yet exact, albeit very close. If we have integer edge capacities, we know that there exists an exact integer solution to the maximum flow problem. This  solution can be obtained from the approximate solution, which might be fractional. This is done by rounding the fractional solution to an integer solution, and computing the final solution from the approximate solution using a single augmenting path. In this section, we present a fast deterministic flow rounding algorithm for the congested clique.

As is standard practice, we use the flow rounding algorithm of Cohen~\cite{Cohen95}. This algorithm uses the decomposition of an unweighted undirected Eulerian graph into directed cycles as a subroutine. Although the original result is in a sequential setting, we can phrase the result more generally as follows and provide the corresponding algorithm in \autoref{app:flow_rounding}. 

\begin{theorem}[Proposition 5.3 of \cite{Cohen95}]\label{thm:cohen_rounding}
Let $G=(V,E)$ be a graph on $n$ vertices and $m$ edges. Let $f\colon E\to \R_{\geq 0}$ be an $s$-$t$ flow function, and $\Delta$ be a real value such that $1/\Delta$ is a power of $2$ and $f(e)$ is an integer multiple of $\Delta$ for every $e\in E$. Suppose there is an algorithm that computes Eulerian orientations in $T$ time. Then there exists an algorithm that rounds $f$ on edge $e\in E$ to $\lfloor f(e) \rfloor$ or $\lceil f(e) \rceil$ such that the resulting flow has total flow value not less than $f$. If the total flow $f$ is integral, and there is an integral cost function $c\colon E \to \Z_{\geq 0}$, then there exists an algorithm that rounds $f$ on edge $e\in E$ to $\lfloor f(e) \rfloor$ or $\lceil f(e) \rceil$ such that the resulting flow has total flow value not less than $f$, and the resulting total cost is no more than that of $f$.

Both algorithms run in time $O(T\log(1/\Delta))$. 
\end{theorem}

In the congested clique, we have the following result for computing Eulerian orientations. 

\ThmEulerianOrientation*
\begin{proof}
    The algorithm consists of the following steps. 
\begin{enumerate}
	\item Since the graph is Eulerian, each node has an even number of neighbors. Internally, each node pairs its neighbors to create cycles. This means we have obtained a decomposition into cycles that of which each node only has local knowledge -- which is obvious, as there has been no communication yet. For each of the cycles a node is part of, it keeps track of a status (active/inactive), initially set to active. \label{step:setup}
	\item For $i=1,\dots \log n$ do: \label{step:iteration}
	\begin{enumerate} 
		\item For each cycle, mark at most half of the vertices, such that there never more than 3 unmarked consecutive vertices, for a procedure see below.   \label{step:sampling}
		\item Each active node begins to send its ID along the cycle into both directions. We do this for $4$ rounds, after which the message has arrived at the next active node. (In the case of a `short' cycle, the message will arrive at the initial node itself.)\\
	Note that, for iterations with $i\geq 2$, the links on which the messages are sent are no longer necessarily disjoint. This is where we use the power of the congested clique model: in a constant number of rounds, these messages can still be delivered to the recipients by using a routing algorithm. Each node is part of at most $n$ cycles, so in any point in time receives and sends at most $n$ messages, which gives a total of $O(n^2)$ messages. Hence we can use the routing algorithm of Lenzen~\cite{Lenzen13} to deliver all messages in at most 16 rounds.   \label{step:find_neighbors} 
		\item Now on each cycle, each active node knows its `neighboring' active nodes. We consider the created cycles of active nodes for the next iteration. \label{step:next_iteration}
	\end{enumerate}
	\item After the $\log n$ iterations, each cycle has at most $O(n\cdot (1/2)^{\log n})=O(1)$ active nodes remaining. In a constant number of rounds, these nodes pick a leader by simply traversing the cycle and finding the node with the highest ID. Each leader arbitrarily picks a direction for its cycle.  \label{step:leaderpicking}
	\item We reverse steps \ref{step:leaderpicking} and \ref{step:iteration} to communicate the chosen orientation to the whole cycle. \label{step:reverse}
\end{enumerate}

Now let us look at how to implement step~\ref{step:sampling}. 
First, we find a maximal matching in $O(\log^* n)$ rounds. This is done by finding a 3-coloring of the cycle in $O(\log^* n)$ rounds \cite{CV86,GPS87}. From this we easily find a maximal matching in a constant number of rounds. 

Next, for each edge of the matching, we mark the vertex with the highest ID. Two edges from the matching can have at most 1 unmatched vertex in between by maximality of the matching, so in total there are at most 3 consecutive unmarked vertices. Moreover, as the marked vertices are half the vertices of the matching, there are at most $n/2$. 

By using the efficient routing techniques of Lenzen~\cite{Lenzen13}, we can simulate this in the congested clique in $O(\log^* n)$ rounds for all cycles simultaneously, even if the cycles in question do not consist of edges of the input graph. 

We see that step~\ref{step:setup} takes no rounds, as this is an internal action. As just shown, step~\ref{step:sampling} takes $O(\log^* n)$ rounds. Step~\ref{step:next_iteration} is again done internally, and step~\ref{step:find_neighbors} is done in a constant number of rounds. In total, this means step~\ref{step:iteration} takes $O(\log n\log^*n)$ rounds. Step~\ref{step:leaderpicking} takes a constant number of rounds, and step~\ref{step:reverse} again takes $O(\log n)$ rounds. So in total the algorithm takes $O(\log n\log^* n)$ rounds. 
\end{proof}

Note that for a randomized algorithm we could replace step~\ref{step:sampling} in \autoref{thm:EulerianOrientation} by randomly sampling each node with constant probability to remove the $\log^* n$-factor in the running time. 

Combining \autoref{thm:cohen_rounding} with \autoref{thm:EulerianOrientation} gives the following result.

\begin{lemma}\label{lm:flowrounding_detCC}
Let $G=(V,E)$ be a graph on $n$ vertices and $m$ edges. Let $f\colon E\to \R_{\geq 0}$ be an $s$-$t$ flow function, and $\Delta$ be a real value such that $1/\Delta$ is a power of $2$ and $f(e)$ is an integer multiple of $\Delta$ for every $e\in E$. There exists an algorithm that rounds $f$ on edge $e\in E$ to $\lfloor f(e) \rfloor$ or $\lceil f(e) \rceil$ such that the resulting flow has total flow value not less than $f$. If the total flow $f$ is integral, and there is an integral cost function $c\colon E \to \Z_{\geq 0}$, then there exists an algorithm that rounds $f$ on edge $e\in E$ to $\lfloor f(e) \rfloor$ or $\lceil f(e) \rceil$ such that the resulting flow has total flow value not less than $f$, and the resulting total cost is no more than that of $f$. Both algorithm run in $O(\log n\log^* n\log(1/\Delta))$ rounds deterministically in the congested clique. 
\end{lemma}
\begin{proof}
    We provide the algorithm \textsc{FlowRounding} solving both cases in \autoref{app:flow_rounding}, which uses the algorithm of \autoref{thm:EulerianOrientation} as a subroutine. For correctness, we refer to Cohen~\cite{Cohen95}. 

    Assuming we can traverse cycles such that the sum of costs on forward edges is no more than the sum of costs on backward edges, \autoref{line:rounding_costs}, the number of rounds follows directly: there are $\log(1/\Delta)$ iterations, where each iterations consists of finding Eulerian orientations and some internal operations. Using that Eulerian orientations can be computed in $O(\log n\log^* n)$ rounds by \autoref{thm:EulerianOrientation} we obtain the result. 

    It remains to show that we can traverse cycles such that the sum of costs on forward edges is no more than the sum of costs on backward edges. This is a small adaptation from \autoref{thm:EulerianOrientation}. Whenever a set of active nodes is selected, we set the cost of the edges between the active nodes to be equal to the sum of the edges in between, while adding a minus sign for backward edges. This means that in the end the leader knows the cost of traversing the cycle in either direction and can choose the direction such that the sum of costs on forward edges is no more than the sum of costs on backward edges.
\end{proof}

\section{Maximum Flow}\label{sc:flow}
In this section, we use our deterministic Laplacian solver and deterministic flow rounding algorithms to solve the maximum flow problem in the congested clique. Hereto, we show that our algorithms allow for an efficient implementation of M{\k{a}}dry's interior point method for computing maximum flows~\cite{Madry16}, hence for correctness of the algorithm we refer to~\cite{Madry16}.

\ThmFlow*
\begin{proof}
We provide the algorithm in \autoref{sc:app_maxflow}. A high-level overview of this algorithm is as follows. We run $\tilde O(m^{3/7}U^{1/7})$ iterations of augmenting flows, where $U$ is the maximum weight. Each iteration consists of solving an electrical flow instance, which corresponds to solving a Laplacian system. The resulting flow after these iterations is fractional, and almost optimal: its value is at most one less then the optimal flow. Rounding this fractional flow and then finding an augmenting path gives us the exact maximum flow. 

Now, let us consider these steps in the congested clique. The algorithm uses the subroutines \textsc{Augmentation}, \textsc{Fixing}, \textsc{Boosting}, and \textsc{FlowRounding}.  \textsc{Augmentation}, \autoref{alg:augmentation}, consists of some operations that can be done in a constant number of rounds and a Laplacian solve (\autoref{line:aug_laplacian}). The latter can be performed in $n^{o(1)}$ rounds by \autoref{thm:Laplacian_detCC}, since we need precision $\Omega(1/\poly(m))$~\cite{Madry16}. Similarly, \textsc{Fixing}, \autoref{alg:fixing}, consists of operations that be done in a constant number of rounds and a Laplacian solve (\autoref{line:fix_laplacian}). \textsc{Boosting}, \autoref{alg:boosting}, consists of only updating variables and can be done in a constant number of rounds. 

Now \textsc{Max Flow}, \autoref{alg:madry}, starts with preconditioning and initialization. Both can be done in a constant number of rounds. The progress step starts with a single call to \textsc{Augmentation} and \textsc{Fixing}, done in $n^{o(1)}$ rounds. Then the for-loop of \autoref{line:flow_forloop} consists of $100\cdot\frac{1}{\widehat{\delta}}\cdot\log{U}=\tilde O(m^{3/7}U^{1/7})$ iterations. Each iteration consists of either a call to \textsc{Augmentation} and \textsc{Fixing}, or a call to \textsc{Boosting}. So in total the for-loop takes $m^{3/7+o(1)}U^{1/7}$ rounds. 

Next, on \autoref{line:flow_round}, we round the flow. We know it is enough to maintain flows with precision $O(1/m)$, so we use \autoref{lm:flowrounding_detCC} with $\Delta = O(1/m)$, and hence rounding takes $O(\log^* n\log^2 n)$ rounds.

Finally, we need to find augmenting paths. The while-loop of \autoref{line:flow_augmenting} actually only needs one iteration~\cite{Madry16}, and this augmenting path can be found by solving (approximate) directed SSSP on the residual graph. We do this in $O(n^{0.158})$ rounds using the algorithm of Censor-Hillel, Kaski, Korhonen, Lenzen, Paz, and Suomela~\cite{Censor-HillelKK19} for a $(1+o(1))$-approximation of weighted directed APSP.

We obtain a total number of rounds of $m^{3/7+o(1)}U^{1/7}$ to solve the exact maximum flow problem in the congested clique. 
\end{proof}

\section{Unit Capacity Minimum Cost Flow}\label{sc:mincostflow}
Similar to \autoref{sc:flow}, we use our Laplacian solver and flow rounding algorithm as subroutines to solve the minimum cost flow problem. We show that these tools allow for an efficient implementation of the algorithm of Cohen, M\k{a}dry, Sankowski, and Vladu~\cite{CMSV17} for unit capacity minimum cost flow. For correctness, we refer to~\cite{CMSV17}, since we give a congested clique implementation of their algoruithm.

\ThmMinCostFlow*
\begin{proof}
    We provide the algorithm \textsc{MinCostFlow}, \autoref{algo:min-cost-flow-final}, in \autoref{sc:app_mincostflow}. It uses four subroutines: \textsc{Initialization}, \textsc{Perturbation}, \textsc{Progress}, and \textsc{Repairing}. We analyze their running times in the congested clique first. 
    \textsc{Initialization}, \autoref{algo:initializationbipartitegraph_app}, and
    \textsc{Perturbation}, \autoref{perturbation}, both consist of some operations that can be done in a constant number of rounds. 
    \textsc{Progress}, \autoref{progress}, consists of some operations that can be done in a constant number of rounds and two Laplacian solves. The latter can be done in $n^{o(1)}$ rounds by \autoref{thm:Laplacian_detCC}, since we need precision $\Omega(1/\poly(m))$~\cite{CMSV17}. 
    \textsc{Repairing}, \autoref{repairing}, first consists of some operations that can be done in a constant number of rounds and a call to \textsc{FlowRounding}. The latter takes $O(\log^2 n\log^* n)$ rounds by applying \autoref{lm:flowrounding_detCC} with $\Delta=O(1/m)$, which is sufficient by~\cite{CMSV17}. Then the algorithm proceeds by a for-loop, \autoref{line:repairing_forloop}, which consists of $\tilde O(m^{3/7})$ iterations. In each iteration, there are some operations that can be performed in a constant number of rounds, a shortest path computation, \autoref{line:repairing_SP}, and another for-loop, \autoref{line:repairing_for2}. For the shortest path computation, approximations suffice~\cite{CMSV17}, hence it can be done in $O(n^{0.158})$ rounds using the algorithm of Censor-Hillel, Kaski, Korhonen, Lenzen, Paz, and Suomela~\cite{Censor-HillelKK19} for a $(1+o(1))$-approximation of weighted directed APSP. The for-loop can be executed internally for all relevant nodes, so takes no rounds. In total this gives $\tilde O(m^{3/7}n^{0.158})$ rounds. 

    Next we analyze \textsc{MinCostFlow}, \autoref{algo:min-cost-flow-final}. The first two lines are the initialization, which takes a constant number of rounds. Then there is a for-loop, \autoref{line:mincostflow_for1}, with $c_Tm^{1/2-3\eta}=O(\log^{4/3} W m^{2/7})$ iterations. In each iteration, there are some variable adjustments that can be done in parallel in a constant number of rounds (the for-loop of \autoref{line:mincostflow_for2}), and another for-loop, \autoref{line:mincostflow_for3}. The latter for-loop consists of $m^{1/7}$ iterations, where each iteration consists of a while-loop, \autoref{line:mincostflow_while}, and a call to \textsc{Progress}. The latter takes $n^{o(1)}$ rounds by the above. The while-loop gets satisfied at most $O(m^{3/7}\poly\log W)$ times~\cite{CMSV17} \emph{in total}. Since each iteration of the while-loop consists of a call to \textsc{Perturbation}, this takes $O(m^{3/7}\poly\log W)$ rounds. 
    Finally, we call \textsc{Repairing}, which takes $\tilde O(m^{3/7}n^{0.158})$ rounds. Summing all running times, we obtain $\tilde O(m^{3/7}(n^{0.158}+n^{o(1)}\poly\log W))$ rounds. 
\end{proof}

\subsection*{Acknowledgements}
We thank the anonymous reviewer who pointed out a mistake in the proof of \autoref{cor:laplaciansolving} and showed how to fix it. 

\printbibliography[heading=bibintoc] 

\appendix

\newpage
\section{Flow Rounding Algorithm}\label{app:flow_rounding}
In this section, we give the flow rounding algorithm, as designed by Cohen~\cite{Cohen95}. We use the notation of Forster, Goranci, Liu, Peng, Sun, and Ye~\cite{FGLP+20}, since this suits applying it in a distributed setting. We run the algorithm in the congested clique, and analyze the corresponding running time in \autoref{sc:flow_rounding}.

\begin{algorithm}[ht]
\caption{$\textsc{FlowRounding}(G, s, t, f, c, \Delta)$\label{algo:rounding}}
\KwIn{graph $G=(V,E)$; source $s\in V$ and sink $t\in V$; flow function $f\colon E \to \R_{\geq 0}$; value $\Delta\in \R_{\geq 0}$ such that $1/\Delta$ is a power of $2$ and $f(e)$ is an integral multiple of $\Delta$ for each  $e\in E$; (optional) cost function $c\colon E \to \Z_{\geq 0}$;}
	\If {the total flow of $f$ is not integral} {
		Add an edge from $t$ to $s$ with flow value the same as total flow.\\
	}
	\While{$\Delta < 1$} {\label{ln:FlowRoundingWhile}
		$ E' \leftarrow \{(u, v) \in  E: f(u, v)/ \Delta \text{ is } odd\}$\\
		Find an Eulerian partition of $ E'$ (ignoring the directions of the edges)\\
		\For{every cycle of the Eulerian partition of $ E'$} {

			\If {cycle contains the edge $(t, s)$} {
			Traverse the cycle such that edge $(t, s)$ is a forward edge.
			}
			\ElseIf {cost function $c$ exists} {
				Traverse the cycle such that the sum of costs on forward edges is no more than the sum of costs on backward edges.\label{line:rounding_costs}
			}
			\Else{
				Traverse the cycle arbitrarily.
			}
		}
		\For{every edge $(u, v)$ $ E'$} {
			\If {$(u, v)$ is a forward edge w.r.t.\ the traversal of the path containing $(u, v)$} {
				$f(u, v) \leftarrow f(u, v) + \Delta$\\
			}
			\Else {$f(u, v) \leftarrow f(u, v) - \Delta$\\}

		}
		$\Delta \leftarrow 2\Delta$\\
	}
	\Return $f$.

\end{algorithm}

\newpage
\section{Maximum Flow Algorithm}\label{sc:app_maxflow}
The algorithms presented in this section are designed by M\k{a}dry~\cite{Madry16}, who proves their correctness. We give the algorithms as they appeared in~\cite{FGLP+20}, since there they are already phrased for a distributed setting -- in their case, this was the CONGEST model. We run the algorithm in the congested clique, and analyze the corresponding running time in \autoref{sc:flow}. 

\begin{algorithm}[H]
\KwIn{directed graph $G_0=(V, E_0,\vec{u})$ with each $e\in E_0$ having two non-negative integer capacities $u_{e}^{-}$ and $u_{e}^{+}$; $|V|={n}$ and $|E_0|={m}$; source $s$ and sink $t$; the largest integer capacity $U$; target flow value $F\ge 0$;}
\caption{\textsc{MaxFlow}($G_0$, $s$, $t$, $U$, $F$)}\label{alg:madry}
\tcc{Preconditioning Edges}
Add $m$ undirected edges $(t, s)$ with forward and backward capacities $2U$ to $G_0$\;
\tcc{Initialization}
\For{each $e=(u,v)\in E_0$}{
replace $e$ by three undirected edges $(u,v)$, $(s,v)$ and $(u,t)$ whose capacities are $u_{e}$\;}
Let the new graph be $G=(V,E)$\;
Initialize $\vec{f}\leftarrow\vec{0}$ and $\vec{y}\leftarrow\vec{0}$\;
\tcc{Progress Step}
$\vec{\widetilde{f}}, \vec{\widehat{f}}, \vec{\widehat{y}}\leftarrow\textsc{Augmentation}(G, s, t, F)$\;
Compute $\vec \rho$ by letting $\rho_{e}\leftarrow\frac{\widetilde{f}_{e}}{\min\{u_{e}^{+}-f_{e}, u_{e}^{-}+f_{e}\}}$\;
$\vec{f}, \vec{y}\leftarrow\textsc{Fixing}\left(G, \vec{\widehat{f}}, \vec{\widehat{y}}\right)$\;
$\eta\leftarrow\frac{1}{14}-\frac{1}{7}\log_{m}{U}-O(\log_m\log(mU))$, 
$\widehat{\delta}\leftarrow\frac{1}{m^{\frac{1}{2}-\eta}}$\;
\For{$t=1$ \KwTo $100\cdot\frac{1}{\widehat{\delta}}\cdot\log{U}$\label{line:flow_forloop}}{
\eIf{$\left\|\vec{\rho}\right\|_{3}\le\frac{m^{\frac{1}{2}-\eta}}{33(1-\alpha)}$}
{$\delta\leftarrow\frac{1}{33(1-\alpha)\left\|\vec{\rho}\right\|_{3}}$\;
$\vec{\widetilde{f}}, \vec{\widehat{f}}, \vec{\widehat{y}}\leftarrow\textsc{Augmentation}(G, s, t, F)$\;
Compute $\vec \rho$ by letting
$\rho_{e}\leftarrow\frac{\widetilde{f}_{e}}{\min\{u_{e}^{+}-f_{e}, u_{e}^{-}+f_{e}\}}$\;
$\vec{f}, \vec{y}\leftarrow\textsc{Fixing}\left(G, \vec{\widehat{f}}, \vec{\widehat{y}}\right)$\;
}
{let $S^{*}$ be the edge set that contains the $m^{4\eta}$ edges with the largest $|\rho_{e}|$\;
$G\leftarrow\textsc{Boosting}\left(G, S^{*}, U, \vec{f}, \vec{y}\right)$\;
}
}
$\vec{f}\leftarrow\textsc{FlowRounding}(G,\vec{f}, s, t)$\label{line:flow_round}\;
\While{there is an augmenting path from $s$ to $t$ with respect to $\vec f$ for $G$\label{line:flow_augmenting}} {
augment an augmenting path for $\vec f$\;
}
\end{algorithm}

\newpage

\vspace{1cm} \begin{algorithm}[H]
\caption{\textsc{Augmentation}($G$, $s$, $t$, $F$)}
\label{alg:augmentation}
\tcc{Augmentation Step}
For each $e\in E$, let $r_{e}\leftarrow\frac{1}{(u_{e}^{+}-f_{e})^{2}}+\frac{1}{(u_{e}^{-}+f_{e})^{2}}$ and $w_{e}\leftarrow\frac{1}{r_{e}}$\;
Solve Laplacian linear system $\mL(G)\vec{\widetilde{\phi}}=F\cdot\vec{\chi}_{s,t}$ where $\vec{\chi}_{s,t}$ is the vector whose entry is $-1$ (resp. $1$) at vertex $s$ (resp. $t$) and $0$ otherwise\label{line:aug_laplacian}\;
For each $e=(u,v)\in E$, let $\widetilde{f}_{e}\leftarrow\frac{\widetilde{\phi}_{v}-\widetilde{\phi}_{u}}{r_{e}}$ and $\widehat{f}_{e}\leftarrow f_{e}+\delta\widetilde{f}_{e}$\;
For each $v\in V$, let $\widehat{y}_{v}\leftarrow y_{v}+\delta\widetilde{\phi}_{v}$\;
\Return $\vec{\widetilde{f}}$, $\vec{\widehat{f}}$, $\vec{\widehat{y}}$\;
\end{algorithm}

\vspace{2cm} \begin{algorithm}[H]
\caption{\textsc{Fixing}$\left(G, \vec{\widehat{f}}, \vec{\widehat{y}}\right)$}
\label{alg:fixing}
\tcc{Fixing Step}
For each $e=(u,v)\in E$, let $r_{e}\leftarrow\frac{1}{\left(u_{e}^{+}-\widehat{f}_{e}\right)^{2}}+\frac{1}{\left(u_{e}^{-}+\widehat{f}_{e}\right)^{2}}$, $w_{e}\leftarrow\frac{1}{r_{e}}$ and
$\theta_{e}\leftarrow w_{e}\left[\left(\widehat{y}_{v}-\widehat{y}_{u}\right)-\left(\frac{1}{u_{e}^{+}-\widehat{f}_{e}}-\frac{1}{u_{e}^{-}-\widehat{f}_{e}}\right)\right]$\;
$\vec{f'}\leftarrow\vec{\widehat{f}}+\vec{\theta}$\;
Let $\vec{\widehat{\delta}}$ be $\vec{\theta}$'s residue vector\;
For each $e\in E$, let $r_{e}\leftarrow\frac{1}{(u_{e}^{+}-f'_{e})^{2}}+\frac{1}{(u_{e}^{-}+f'_{e})^{2}}$ and $w_{e}\leftarrow\frac{1}{r_{e}}$\;
Solve Laplacian linear system $\mL(G)\vec{\phi'}=-\vec{\widehat{\delta}}$\label{line:fix_laplacian}\;
For each $e=(u,v)\in E$, let $\theta_{e}'\leftarrow\frac{\phi_{v}'-\phi_{u}'}{r_{e}}$\;
For each $e\in E$, $f_{e}\leftarrow f'_{e}+\theta_{e}'$\;
For each vertex $v\in V$, $y_{v}\leftarrow \widehat{y}_{v}+\phi_{v}'$\;
\Return $\vec{f}$, $\vec{y}$\;
\end{algorithm}

\newpage

\begin{algorithm}[H]
\caption{\textsc{Boosting}$\left(G, S^{*}, U, \vec{f}, \vec{y}\right)$}
\label{alg:boosting}
\tcc{Boosting Step}
\For{each edge $e=(u,v)\in S^{*}$}{
$\beta(e)\leftarrow 2+\lceil{\frac{2U}{\min\{u_{e}^{+}-f_{e}, u_{e}^{-}+f_{e}\}}}\rceil$\;
replace $e$ with path $u\leadsto v$ that consists of $\beta(e)$ edges $e_{1}, \cdots, e_{\beta(e)}$ oriented towards $v$ and $\beta(e)+1$ vertices $v_{0}=u, v_{1},\cdots, v_{\beta(e)-1}, v_{\beta(e)}=v$\;
$e_{1}, e_{2}\leftarrow e$\;
for $3\le i\le\beta(e)$, let $u_{e_{i}}^{+}\leftarrow +\infty$ and $u_{e_{i}}^{-}\leftarrow\left(\frac{1}{u_{e}^{+}-f_{e}}-\frac{1}{u_{e}^{-}+f_{e}}\right)^{-1}(\beta(e)-2)-f_{e}$\;
for each $1\le i\le\beta(e)$, let $f_{e_{i}}\leftarrow f_{e}$\;
$y_{v_{0}}\leftarrow y_{u}$\;
$y_{v_{\beta(e)}}\leftarrow y_{v}$\;
$y_{v_{1}}\leftarrow y_{v}$\;
$y_{v_{2}}\leftarrow y_{v}+\frac{1}{u_{e}^{+}-f_{e}}-\frac{1}{u_{e}^{-}+f_{e}}$\;
    for $3\le i\le\beta(e)$, set $y_{v_{3}},\cdots,y_{v_{\beta(e)-1}}$ such that $y_{v_{i}}-y_{v_{i-1}}=-\frac{1}{\beta(e)-2}\left(\frac{1}{u_{e}^{+}-f_{e}}-\frac{1}{u_{e}^{-}+f_{e}}\right)$\;
}
Update $G$\;
\end{algorithm}

\section{Unit Capacity Minimum Cost Flow Algorithm}\label{sc:app_mincostflow}
The algorithms presented in this section are designed by Cohen, M\k{a}dry, Sankowski, and Vladu~\cite{CMSV17}, who prove their correctness. We give the algorithms as they appeared in~\cite{FGLP+20}, since there they are already phrased for a distributed setting -- in their case, this was the CONGEST model. We run the algorithm in the congested clique, and analyze the corresponding running time in \autoref{sc:mincostflow}.

\begin{algorithm}[H]
\KwIn{directed graph ${G_0}=({V_0},{E_0},\vec{c}_0)$ with each edge having unit capacity and cost $\vec{c}_0$; $|{V_0}|={n}$ and $|{E_0}|={m}$; integral demand vector $\vec{\sigma}$; the absolute maximum cost $W$;}
\caption{\textsc{MinCostFlow}($G$, $\vec{\sigma}$, $W$)}
\label{algo:min-cost-flow-final}
$G = (P\cup Q, E), \vec b, \vec{f}, \vec{y}, \vec{s}, \vec{\nu}, \widehat{\mu}$, $c_{\rho}, c_{T}, \eta\leftarrow \textsc{Initialization}(G_0, \vec \sigma)$\;
Add a new vertex $v_0$ and undirected edges $(v_0, v)$ for every $v\in P$ to $G$\;
\For{$i=1$ \KwTo $c_{T}\cdot m^{1/2-3\eta}$\label{line:mincostflow_for1}
}{
\For{each $v\in P$\label{line:mincostflow_for2}}
{
set resistance of edge $(v_0, v)$ for each $v\in P$ to be $r_{v_{0}v}\leftarrow\frac{m^{1+2\eta}}{a(v)}$, where $a(v)\leftarrow\sum_{u\in Q, e=(v,u)\in E}\nu_e+\nu_{\overline{e}}$; \tcc{$\overline{e}=(\overline{v},u)$ is $e$'s partner edge that is the unique edge sharing one common vertex from $Q$.}
}
\For{$j=1$ \KwTo $m^{2\eta}$\label{line:mincostflow_for3}}{
\While{$\left\|\vec{\rho}\right\|_{\vec{\nu},3}>c_{\rho}\cdot m^{1/2-\eta}$\label{line:mincostflow_while}}{
	$\vec{\rho}, \vec{y}, \vec{s}, \vec{\nu}\leftarrow \textsc{Perturbation}(G, \vec{\rho}, \vec{f}, \vec{y}, \vec{s}, \vec{\nu})$\;}
	$\vec{f}, \vec{s}, \vec{\rho}, \widehat{\mu} \leftarrow \textsc{Progress}(G, \vec{\sigma}, \vec{f}, \vec{\nu})$\;
}}
\textsc{Repairing}($G$, $\vec{f}$, $\vec{y}$)\;
\end{algorithm}

\newpage
\begin{algorithm}[H]
\caption{\textsc{Initialization}($G$, $\vec{\sigma}$) \label{algo:initializationbipartitegraph_app}}
Create a new vertex $v_{aux}$ with $\sigma(v_{aux})=0$\;
\For{each $v\in{V_0}$}{
$t(v)\leftarrow\sigma(v)+\frac{1}{2}\mathrm{deg}_{in}^{{G_0}}(v)-\frac{1}{2}\mathrm{deg}_{out}^{{G_0}}(v)$\;
\lIf{$t(v)>0$}{
construct $2t(v)$ parallel edges $(v, v_{aux})$ with costs $\left\|\vec{{c}}_0\right\|_1$
}
\uElseIf{$t(v)<0$}{
construct $|2t(v)|$ parallel edges $(v_{aux}, v)$ with costs $\left\|\vec{{c}}_0\right\|_1$\;
}}
Let the new graph be $G_1=(V_1, E_1, \vec{c}_1)$\;
Initialize the bipartite graph $G=(P\cup Q, E, \vec{c})$ with $E\leftarrow\emptyset
$, $P\leftarrow V_1$ and $Q\leftarrow\{e_{uv}\mid(u,v)\in E_1\}$ where $e_{uv}$ is a vertex corresponding to edge $(u,v)\in E_1$\;

\For {each $(u,v)\in E_1$} {
let $E\leftarrow E\cup\{(u, e_{uv}), (v, e_{uv})\}$ with $c(u, e_{uv})=c_1(u, v)$ and $c(v, e_{uv})=0$, and set $b(u)\leftarrow\sigma(u)+\mathrm{deg}_{in}^{{G_1}}(u)$, $b(v)\leftarrow\sigma(v)+\mathrm{deg}_{in}^{{G_1}}(v)$ and $b(e_{uv})\leftarrow 1$\;
}
For each $v\in P$, set $y_v\leftarrow\left\|\vec{c}\right\|_{\infty}$, and for each $v\notin P$, set $y_v\leftarrow 0$\;
For each $e=(u,v)\in E$, set $f_e\leftarrow\frac{1}{2}$, $s_e\leftarrow c_e+y_u-y_v$ and $\nu_e\leftarrow\frac{s_e}{2\left\|\vec{c}\right\|_{\infty}}$\;
Set $\widehat{\mu}\leftarrow\left\|\vec{c}\right\|_{\infty}$,
$c_{\rho}\leftarrow 400\sqrt{3}\cdot\log^{1/3}{W}$,
$c_{T}\leftarrow 3c_{\rho}\log{W}$ and 
$\eta\leftarrow\frac{1}{14}$\;
\Return $G$, $\vec b$, $\vec{f}$, $\vec{y}$, $\vec{s}$, $\vec{\nu}$, $\widehat{\mu}$, $c_{\rho}$, $c_{T}$ and $\eta$\;
\end{algorithm}

\begin{algorithm}[H]
\caption{\textsc{Perturbation}($G$, $\vec{\rho}$, $\vec{f}$, $\vec{y}$, $\vec{s}$, $\vec{\nu}$)}
\label{perturbation}
\For{each $v\in Q$}{
let $e=(u,v)$ and $\overline{e}=(\overline{u}, v)$\;
$y_v\leftarrow y_v-s_e$\;
$\nu_e\leftarrow 2\nu_e$\;
$\nu_{\overline{e}}\leftarrow\nu_{\overline{e}}+\frac{\nu_{e}f_{\overline{e}}}{f_e}$\;
}
\end{algorithm}

\begin{algorithm}[H]
\caption{\textsc{Progress}($G$, $\vec{\sigma}$, $\vec{f}$, $\vec{\nu}$)}
\label{progress}
For each $e\in E$, let $r_{e}\leftarrow \frac{\nu_{e}}{f_{e}^{2}}$\;
Solve Laplacian linear system $\mL(G)\vec{\widehat{\phi}}=\vec{\sigma}$\;
For each $e=(u,v)\in E$, let $\widehat{f}_{e}\leftarrow\frac{\widehat{\phi}_{v}-\widehat{\phi}_{u}}{r_{e}}$
and $\rho_{e}\leftarrow\frac{|\widehat{f}_{e}|}{f_{e}}$\;
$\delta\leftarrow\min\left\{\frac{1}{8\left\|\vec{\rho}\right\|_{\vec{\nu},4}},\frac{1}{8}\right\}$\;
Update $f_{e}'\leftarrow(1-\delta)f_{e}+\delta\widehat{f}_{e}$
and
$s_{e}'\leftarrow s_{e}-\frac{\delta}{1-\delta}(\widehat{\phi}_{v}-\widehat{\phi}_{u})$\;
For each $e\in E$, let $f_{e}^{\#}\leftarrow\frac{(1-\delta )f_{e}s_{e}}{s_{e}'}$\;
Obtain the flow vector $\vec{\sigma'}$ corresponding to the residue $\vec{f'}-\vec{f}^{\#}$\;
For each $e\in E$, let $r_{e}\leftarrow\frac{s_{e}'^{2}}{(1-\delta)f_{e}s_{e}}$\;
Solve Laplacian linear system $\mL(G)\vec{\widetilde{\phi}}=\vec{\sigma'}$\;
For each $e=(u,v)\in E$, let $\widetilde{f}_{e}\leftarrow\frac{\widetilde{\phi}_{v}-\widetilde{\phi}_{u}}{r_{e}}$\;
Update $f_{e}\leftarrow f_{e}^{\#}+\widetilde{f}_{e}$ and $s_{e}\leftarrow s_{e}'-\frac{s_{e}'\widetilde{f}_{e}}{f_{e}^{\#}}$\;
\end{algorithm}

\newpage
\begin{algorithm}
\caption{\textsc{Repairing}($G$, $\vec{f}$, $\vec{y}$)}
\label{repairing}
Let $\vec{b}^+$ be the demand vector corresponding to the current flow $\vec{f}$\;
For each $v\in P\cup Q$, set $b_{v}^{\le}\leftarrow\min(b_{v},b_{v}^{+})$\;
For each $v\in P\cup Q$, if $f(E(v))>b_{v}^{\le}$, set $\vec f$ on $E(v)$ such that $f(E(v))= b_{v}^{\le}$, and let the resulting vector be $\vec{f}^{\le}$\;
Add source $s$ and sink $t$ to $G$, and 
connect $s$ to each $v\in P$ with $f^{\le}_{sv}\leftarrow f^{\le}(E(v))$, and connect each $v\in Q$ to $t$ with $f^\le_{vt}\leftarrow f^{\le}(E(v))$ in $G$\;
$\vec M\leftarrow\textsc{FlowRounding}(G, \vec{f}^{\le}, s, t)$\;
Remove $s, t$ and related coordinates on $\vec f^\le$ and $\vec M$ from $G$\;
\For{$i=1$ to $\widetilde{O}(m^{3/7})$\label{line:repairing_forloop}}{
construct graph $\overrightarrow{G}_{M}=(P\cup Q, E_M, \widetilde{c}_M)$ using $G$, $\vec{M}$ and $\vec{\widetilde{c}}$ such that for each $e=(u,v)\in E$, $\widetilde{c}_e=c_e-y_u-y_v$ and 
$E_M=\{(u, v)\in E\mid u\in P, v\in Q\}\cup\{(u,v)\mid u\in Q, v\in P, M_{uv}\neq 0\}$, $\widetilde{c}_M(u,v)=\left\{\begin{array}{ll}
     \widetilde{c}_{uv},&u\in P, v\in Q  \\
     -\widetilde{c}_{uv},&u\in Q, v\in P 
\end{array}\right.$\;
set $F_{M}\leftarrow\{v\in P\cup Q\mid M(v)<b_{v}\}$\;
compute a shortest path $\pi$ in $\overrightarrow{G}_{M}$ from $P\cap F_{M}$ to $Q\cap F_{M}$\label{line:repairing_SP}\;
\tcc{$\mathcal{D}_{\overrightarrow{G}_{M}}(P,u)$ is the distance from $P$ to $u$ in $\overrightarrow{G}_{M}$}
\tcc{Edges that are reachable in $\overrightarrow{G}_M$ from $P \cap F_M$ have non-negative weights $\widetilde{c}_e$}
\For{$u\in P\cup Q$\label{line:repairing_for2}}{
\If{$u$ can be reached from $P$ in $\overrightarrow{G}_{M}$}{
\eIf{$u\in P$}{$y_{u}\leftarrow y_{u}-\mathcal{D}_{\overrightarrow{G}_{M}}(P,u)$\;
}{$y_{u}\leftarrow y_{u}+\mathcal{D}_{\overrightarrow{G}_{M}}(P,u)$\;}
}}
augment $\vec{M}$ using the augmenting path $\pi$\;
}
\Return $\vec{M}$\;
\end{algorithm}

\end{document}